\documentclass[aps,pre,reprint,superscriptaddress,longbibliography,showpacs]{revtex4-1}

\usepackage{amsmath,amssymb,graphicx,amsthm}
\usepackage{tikz}
\usetikzlibrary{decorations.pathreplacing}
\usepackage{nicefrac}

\newtheorem{theorem}{Theorem}
\newtheorem{corollary}{Corollary}
\begin{document}


\title{Linked cluster expansion on trees} \author{Deepak Iyer}
\affiliation{Department of Physics \& Astronomy, Bucknell University,
  1 Dent Dr, Lewisburg, PA 17837, USA} \author{Yuyi Wan}
\affiliation{Department of Physics and Astronomy, University of Notre
  Dame, 225 Nieuwland Science Hall, Notre Dame, IN 46556, USA}
\date{\today}

\begin{abstract}
  The linked cluster expansion has been shown to be highly efficient
  in calculating equilibrium and nonequilibrium properties of a
  variety of 1D and 2D classical and quantum lattice models. In this
  article, we extend the linked cluster method to the Cayley tree and
  its boundaryless cousin the Bethe lattice.  We aim to (a) develop
  the linked cluster expansion for these lattices, a novel
  application, and (b) to further understand the surprising
  convergence efficiency of the linked cluster method, as well as its
  limitations. We obtain several key results. First, we show that for
  nearest-neighbor
  Hamiltonians of a specific form, all finite tree-like clusters can be mapped
  to one dimensional finite chains. We then show that the qualitative
  distinction between the Cayley tree and Bethe lattice appears due to
  differing lattice constants that is a result of the Bethe lattice
  being boudaryless. We use these results to obtain the explicit
  closed-form formula for the zero-field susceptibility for the entire
  disordered phase upto the critical point for Bethe lattices of
  arbitrary degree; remarkably, only 1D chain-like clusters
  contribute. We also obtain the exact zero field partition function
  for the Ising model on both trees with only the two smallest
  clusters, similar to the 1D chain. Finally, these results achieve a
  direct comparison between an infinite lattice with a
  non-negligible boundary and one without any boundary, allowing us to show
  that the linked cluster expansion eliminates boundary terms
  at \emph{each order} of the expansion, answering the question about
  its surprising convergence efficiency. We conclude with some
  ramifications of these results, and possible generalizations and
  applications.
\end{abstract}

\maketitle

\section{Introduction}
\label{sec:introduction}

Amongst the wide variety of lattice models used in statistical
mechanics, a small handful are exactly solvable in the thermodynamic
limit\cite{baxter2007}. Examples include the 2D Ising model on a
square lattice, various ``ice'' models, and some quantum models in one
dimension such as the Heisenberg model, the nonlinear Schr\"odinger
equation, and some relativistic models\cite{thacker81,korepin1993}.
Whereas these models have demonstrated surprisingly wide
applicability, in many situations we are forced to lift certain
assumptions, rendering them no longer exactly solvable.  Besides, even
in cases where exact solutions are available, not all physical
quantities can be calculated via closed form expressions, or easily
translated into an experimentally relevant language\cite{korepin1993}.

We are then left with approximate solutions. The success of an
approximation method often has to do with the underlying physics. For
instance, it is notoriously hard to obtain good approximations for
long range order, precisely because in order to capture long range
correlations, one needs large system sizes, and any approximation that
relies on truncating the system size can only give us hints of what
long range order might lie beyond\cite{oitmaa2006}. Perturbative
methods that work directly on infinitely large systems can overcome
this issue, but often do not allow easy access to the parameter
regimes where long range order appears. Indeed, strongly correlated
physics is perhaps the most elusive physics to effectively
model. Similarly, strongly out of equilibrium physics like a quantum
quench is often intractable because simple low-energy approximations
fail\cite{mitra2018}.

The most basic approximation method for a lattice model relies on
studying the properties of small systems as a function of the system
size and carrying out an appropriate scaling, or attempting an
extrapolation from a trend. Simple extrapolations will by definition
fail to capture singularities, which indicate phase transitions, since
these are sudden deviations from the behavior away from the
singularity. Nevertheless, these methods are very effective away from
a phase transition when we are deep in a particular phase.  Other
methods need to and can be employed in combination to recognize where
these phase boundaries lie\cite{baxter2007}. Within such finite size
approximations, the particular statistical ensemble and boundary
conditions used play a strong role in
convergence, with open boundary conditions giving rise to $O(1/N)$ errors, where $N$ is the system size, coming from the boundary of each finite size system\cite{iyer2015,*iyer2015err}.
Further, cluster methods such
as the linked cluster
expansion\cite{sykes66,domb60a,*domb60b,tang_khatami_13} seem to do
even better\cite{iyer2015}. However, reasons for the latter are not
known. In this article, we resolve this outstanding issue by showing that the linked
cluster expansion eliminates the boundary contribution at 
\emph{each order} of the expansion.

The linked cluster expansion has found tremendous
application in classical and quantum systems, especially in recent
studies pertaining to their
dynamics\cite{rigol2014,tang2015a,tang2015b,mallayya2017}, as well as
in studies of disordered or inhomogeneous
systems\cite{tang2015a,tang2015b,mallayya2021,gan2020} (see articles cited in
these references for earlier work) and periodically driven systems using Floquet Hamiltonians\cite{zhang2016}, and has proved to be a remarkably
effective method for approximating lattice models. At an intuitive
level, link cluster expansions of the kind used here operate by
singling out the ``new'' contributions to an extensive quantity at any
stage of the finite size approximation, by effectively canceling out
contributions that are merely appearing from the smaller systems
embedded in a larger system. For example, if we know the physics for a
system size $N_1$, and all of the physics in a system of size $N_2$
appears due to the multiple copies of the smaller system, then at the
higher level, the linked cluster expansion gives us a \emph{zero}
contribution. It appears to be a more efficient method to ``extract
the physics'' at smaller system sizes as has been observed in the
works referenced above.

In this article, we continue exploring the linked cluster expansion
(LCE) in the context of exponentially growing tree lattices, and study
if the above improvements in convergence efficiency prevail.  The
systems so far studied using the linked cluster expansions are regular
lattices, occasionally with disorder.  From the physical standpoint,
tree lattices are unusual given their exponentially growing structure
and in the case of a Cayley tree, the presence of a boundary that has
as many vertices as the entire bulk --- this latter property upends
the common wisdom that in the thermodynamic limit, the boundary does
not significantly contribute to the bulk properties. The Bethe lattice, on the other hand,
does not have a boundary at all, and looks the same from every vertex. Despite these unusual properties, they
have proved to be exceptionally useful lattices to study several
models on. The Ising model on the Bethe lattice bears similar
thermodynamics to the mean field approximation, and has the same
critical exponents\cite{baxter2007}.  More recently, the Bethe lattice
has proved useful in studies of Anderson and many-body
localization\cite{chacra1973,basko2006,savitz2019}. Cayley trees and
Bethe lattices have been extensively studied and find application in a
variety of problems \cite{ostilli2012}. Tree lattices are
often studied using simple finite size approximations or
self-similarity based methods. The latter can be used to provide
implicit expressions for the magnetization of the Ising model in both
ordered and disordered phase, as well as reveal information about the
critical point via a set of exact self-consistent equations that can
be solved numerically to very high accuracy\cite{baxter2007}.

Our goal in this article is twofold -- to study microscopically how
the linked cluster expansion works in the context of a simple nearest
neighbor model on a Cayley tree and its infinite/boundaryless and
rootless sibling the Bethe lattice, and also to understand its surprising convergence
efficiency.   In the following sections, we
systematically develop the linked cluster expansion on trees,
establish equivalences between finite trees and one dimensional
chains, and calculate the exact zero-field partition function for the
Ising model on both types of trees.  We then go on to study the weak
field approximation with a hope to extract the critical temperature
and indeed show that this is possible within the linked cluster
expansion framework, showing a first example of a model where the
$N=2$ system is capable of giving us the critical point and the exact
formula for the susceptibility at zero field. We use these results that allow us to compare
how the linked cluster expansion operates on the Cayley tree and the Bethe lattice to conclude
that the convergence efficiency of the linked cluster expansion is because it
eliminates boundary contributions (known to be the source of poor convergence in systems
with open boundary conditions\cite{iyer2015}) at each stage of the expansion.

\section{Definitions}
\label{sec:defs}

Cayley trees and Bethe lattices are tree graphs, i.e., graphs that are
connected and do not have any loops. In other words, it is not
possible to make a circuit and return to the starting point (vertex)
without retracing one or more edges. The absence of loops is crucial
from a physical standpoint. As an example, the reason the Ising model
on a square lattice differs from the mean-field approximation on a
lattice with the same vertex degree is because of the presence of
loops; without the loops the model can be studied using a
Bethe-Peierls approximation of the appropriate vertex degree, and
gives different critical exponents. Note that a Hamiltonian with next
nearest neighbor hopping or interaction fundamentally destroys the
tree structure by creating loops; we do not consider such models. In
other words, we assume a Hamiltonian that has only nearest neighbor
interactions or hopping.

An $m$-Cayley tree is constructed by starting with one vertex and
drawing $m\geq3$ edges out from it. From each of these new vertices in
the first ``shell'', $m-1$ new edges emerge (for a total of $m$ edges
at each vertex). The Cayley tree therefore grows symmetrically and can
be terminated at any shell. The outermost shell has vertices that are
attached to only one edge each. It is finite, and one can meaningfully
ask a question about the infinite or thermodynamic limit.

An $m$-Bethe lattice on the other hand has no center (root vertex) and
no boundary.  It is a connected graph where \emph{every} vertex is
attached to $m$ edges without creating any loops.  The graph is
therefore entirely self-similar and appears the same from every
vertex.  It is infinite, and there is no meaningful finite subset of
it\cite{ostilli2012}. Nevertheless, as we show below, we can use the
linked cluster expansions that relies on computing properties on
progressively growing finite clusters.  It is critical to note that a
simple finite size extrapolation based on finite clusters is
unreliable and will generally fail on the Bethe lattice since it does
not appropriately account for the absence of a boundary.

Fig.~\ref{fig:cayleybetheexample} shows some illustrations of the two
trees. To simplify the visualization, we choose $m=3$.  In what
follows, we also restrict to $m=3$, noting that all results are
generalizable to arbitrary integer $m\geq 3$ (the $m=2$ case is the
one dimensional chain).
\begin{figure}[h]
  \centering
  \begin{tikzpicture}[scale=0.5]
    \tikzstyle{vertex}=[circle,draw,fill,inner sep=2] \node[vertex]{}
    child[grow=90] {node[vertex]{} child[grow=30] {node[vertex]{}}
      child[grow=150] {node[vertex]{}} } child[grow=210]
    {node[vertex]{} child[grow=150] {node[vertex]{}} child[grow=270]
      {node[vertex]{}} } child[grow=330] {node[vertex]{}
      child[grow=30] {node[vertex]{}} child[grow=270] {node[vertex]{}}
    };
  \end{tikzpicture}
  \hspace{5pt}
  \begin{tikzpicture}[scale=0.5]
    \tikzstyle{vertex}=[circle,draw,fill,inner sep=2] \node[vertex]{}
    child[grow=90] {node[vertex]{} child[grow=30] {node[vertex]{}
        child[grow=330] {node{} edge from parent[dashed]}
        child[grow=90] {node{} edge from parent[dashed]} }
      child[grow=150] {node{} edge from parent[dashed]} }
    child[grow=210] {node[vertex]{} child[grow=150] {node[vertex]{}
        child[grow=90] {node{} edge from parent[dashed]}
        child[grow=210] {node{} edge from parent[dashed]} }
      child[grow=270] {node{} edge from parent[dashed]} }
    child[grow=330] {node[vertex]{} child[grow=30] {node{} edge from
        parent[dashed]} child[grow=270] {node{} edge from
        parent[dashed]} };
  \end{tikzpicture}
  \caption{A finite Cayley tree with $m=3$ and $N=10$ vertices (left),
    and a Bethe lattice with $m=3$ (right).  The dashed lines signify
    the infinite continuation of the tree structure.}
  \label{fig:cayleybetheexample}
\end{figure}
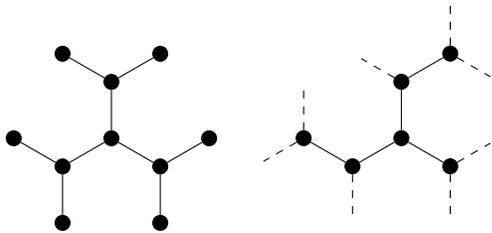

\section{Partition functions on finite tree graphs}
\label{sec:Ztree}

On a finite tree graph like the Cayley tree, for $N$ vertices, there
are always $N-1$ edges, since growing the lattice always involves
adding one or more vertices and an equal number of edges.  Consider a
classical nearest neighbor spin Hamiltonian given by
\begin{equation}
  \label{eq:1}
  H = \sum_{\langle ij\rangle}H_{ij}(s_{i},s_{j})
\end{equation}
The grand partition function is given by a sum over all possible spin
configurations:
\begin{equation}
  \label{eq:2}
  Z = \sum_{\{s_{j}\}} e^{-\beta \sum_{\langle ij\rangle}H_{ij}(s_{i},s_{j})}.
\end{equation}
The linked cluster expansion being a series expansion cannot use the
canonical partition function since the system size, and therefore
total magnetization (or charge or particle number) changes from one
order to the next.  Consider a spin that can take $q$ different
values. In what follows, we restrict ourselves to Hamiltonians of the
form
\begin{equation}
  \label{eq:6}
  H = \sum_{\langle i j \rangle} H_{ij}(|s_i - s_j|).
\end{equation}
The reason for the restriction will become clear in the theorem
below---in short, it ensures that all possible values of $H_{ij}$ can
be obtained by changing only one spin in the pair. For now, we note
that the 1D spin-1/2 Ising model can indeed be cast into the above
form as $H_{ij} = J [(s_i-s_j)^2/2-1]$.  Other examples of
Hamiltonians that have this form include the standard $q$-state Potts
model, given by $H_{ij} = -J\delta_{s_is_j}$, and its cyclic form,
given by $H_{ij} = -J\cos[2\pi(s_i-s_j)/q]$.  The theorem, however,
does not apply to the spin-1 Ising model, for example, which cannot be
cast into this form.


With this constraint on $H$, we show that the partition function for a
given $H$ is identical on all tree graphs with the same number of
vertices $N$.

\begin{theorem}
  \label{thm:tree1d}
  For every finite $m$-tree with a given vertex (spin) configuration
  (or vertex set) $\{v_{i}\}$ and corresponding edge configuration
  (edge set) $\{e_{ij}\}$ obtained from a Hamiltonian
  $H_{ij}(v_{i},v_{j})$ such as in Eq.~\eqref{eq:6}, there exists an
  equivalent finite 1D chain with the same edge set, and therefore the
  same value of the Hamiltonian.
\end{theorem}
\begin{proof}

  We first note that for the Hamiltonians we consider, the value of
  the Hamiltonian depends only on the edge set
  $e_{ij} = H_{ij}(s_i,s_j)$.

  We show the result by constructing the tree with the vertex set
  $\{v_{i}\}$ from a 1D chain by maintaining the same edge set while
  ensuring that there is a one-one mapping between the vertex set of
  the chain and the vertex set of the tree.


  Consider now a specific configuration of spins (vertices)
  $s_{1},\ldots,s_{N}$ of a 1D chain. This corresponds to a specific
  edge configuration $e'_{12},\ldots,e'_{N-1,N}$ given by the
  Hamiltonian $H_{ij}$.  We show that there is another vertex
  configuration $v_{1},\ldots,v_{N}$ on the desired tree graph that
  leads to the same edge configuration (and therefore the same value
  of the Hamiltonian) that obtains from the rearrangement of bonds
  produced by transforming the 1D chain into the tree.  The finite 1D
  chain Hamiltonian is given by
  \begin{equation}
    \label{eq:3}
    H = \sum_{j=1}^{N-1}H_{j,j+1}
  \end{equation}
  We define a ``move'' $M_{i}$ that takes the last available edge from
  the original 1D chain and connects it to the $i$-th vertex,
  retaining the labeling of the original chain. For example, $M_{2}$
  would move the last edge and join it to the second vertex producing
  one degree 3 vertex.  The Hamiltonian after this move becomes
  \begin{equation}
    \label{eq:4}
    H(M_{2}) = \sum_{j=1}^{N-2}H_{j,j+1} + H_{2,N}
  \end{equation}
  Fig.~\ref{fig:move} shows an example of such moves.
  \begin{figure}[h]
    \centering
    \begin{tikzpicture}
      \node[shape=circle,draw=black,fill,inner sep=2pt, label = 1] (a)
      at (1,0) {}; \node[shape=circle,draw=black,fill,inner sep=2pt,
      label = 2] (b) at (2,0) {};
      \node[shape=circle,draw=black,fill,inner sep=2pt, label = 3] (c)
      at (3,0) {}; \node[shape=circle,draw=black,fill,inner sep=2pt,
      label = 4] (d) at (4,0) {};
      \node[shape=circle,draw=black,fill,inner sep=2pt, label = 5] (e)
      at (5,0) {}; \node[shape=circle,draw=black,fill,inner sep=2pt,
      label = 6] (f) at (6,0) {}; \path[-] (a) edge (b); \path[-] (b)
      edge (c); \path[-] (c) edge (d); \path[-] (d) edge (e); \path[-]
      (e) edge (f);
    \end{tikzpicture}\\
    $\downarrow M_2$\\
    \begin{tikzpicture}
      \node[shape=circle,draw=black,fill,inner sep=2pt, label = 1] (a)
      at (1,0) {}; \node[shape=circle,draw=black,fill,inner sep=2pt,
      label = below:2] (b) at (2,0) {};
      \node[shape=circle,draw=black,fill,inner sep=2pt, label = 3] (c)
      at (3,0) {}; \node[shape=circle,draw=black,fill,inner sep=2pt,
      label = 4] (d) at (4,0) {};
      \node[shape=circle,draw=black,fill,inner sep=2pt, label = 5] (e)
      at (5,0) {}; \node[shape=circle,draw=black,fill,inner sep=2pt,
      label = 6] (f) at (2,1) {}; \path[-] (a) edge (b); \path[-] (b)
      edge (c); \path[-] (c) edge (d); \path[-] (d) edge (e); \path[-]
      (b) edge (f);
    \end{tikzpicture}\\
    $\downarrow M_3$\\
    \begin{tikzpicture}
      \node[shape=circle,draw=black,fill,inner sep=2pt, label = 1] (a)
      at (1,0) {}; \node[shape=circle,draw=black,fill,inner sep=2pt,
      label = below:2] (b) at (2,0) {};
      \node[shape=circle,draw=black,fill,inner sep=2pt, label = 3] (c)
      at (3,0) {}; \node[shape=circle,draw=black,fill,inner sep=2pt,
      label = 4] (d) at (4,0) {};
      \node[shape=circle,draw=black,fill,inner sep=2pt, label =
      below:5] (e) at (3,-1) {};
      \node[shape=circle,draw=black,fill,inner sep=2pt, label = 6] (f)
      at (2,1) {}; \path[-] (a) edge (b); \path[-] (b) edge (c);
      \path[-] (c) edge (d); \path[-] (c) edge (e); \path[-] (b) edge
      (f);
    \end{tikzpicture}
    \caption{Example of a series of moves $M_j$ used in Theorem
      \ref{thm:tree1d} on a $N=6$ chain.}
    \label{fig:move}
  \end{figure}
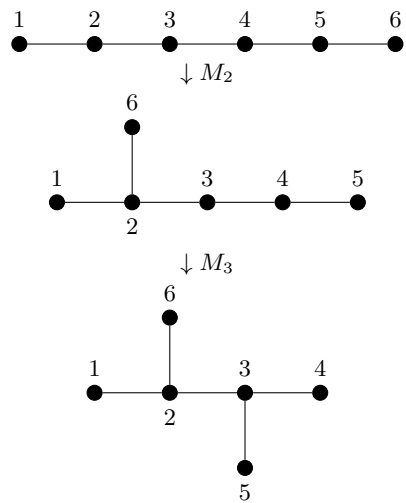
  Any tree can be formed from a sequence of such moves. If the new
  edge $e_{2,N} = e'_{N-1,N}$, then we retain the same edge set. If
  they are not equal, we change $s_{N}\to v_{N}$ such that
  $e_{2,N} = e'_{N-1,N}$. Thus we obtain a new vertex set
  $s_{1},\ldots,s_{N-1},v_{N}$.  For this to work, we require that all
  edge values can be achieved by changing \emph{one} vertex in a
  pair. For the spin-1/2 Ising model, this requirement is clearly
  satisfied. For the $q$-state Potts model, which is a many to two
  mapping, if $e'_{N-1,N} = -J$, then we require $v_N=s_2$ so that
  $e_{2,N} = 1$. This can always be arranged.  For $e'_{N-1,N}=0$, we
  require $v_N\neq s_2$. For $s_1,\ldots,s_{N-1}$ fixed, there are
  $q-1$ equivalent values that $s_N$ can take that result in the same
  value of $H$ on the chain. This is also true on a tree derived from,
  say, $M_2$. Therefore each vertex configuration of the chain can be
  mapped on to a unique vertex configuration of the tree, ensuring
  that we have the same number of configurations with a given
  energy. The argument however fails for the spin-1 Ising
  model. There, if $s_2=0$, then no matter the value of $v_N$, we
  cannot achieve a value different from 0 for $e_{2N}$. This failure
  of the spin-1 Ising has to do with it not satisfying the requirement
  that it have a Hamiltonian of the form in Eq.~\eqref{eq:6}, which
  ensures that all edge values can be achieved by changing one spin of
  a pair, since for $i\in \{1,\ldots,q\}$, $\{|s_i - s_j|\}$ is the
  same set for all values of $j$.

  Any tree can thus be constructed by a sequence of moves $M_{i}$, and
  after every move, we can restore the edge set by making one change
  to the vertex set, ensuring that the mapping is 1-1. In this way,
  every vertex set of the 1D chain goes to an equivalent vertex set of
  any tree with the same edge set. We have therefore established, by
  construction, a one-one mapping between vertex sets of the 1D chain
  and that of any tree that leaves the edge set invariant. This
  implies that there is a mapping between a 1D chain Hamiltonian a
  corresponding tree Hamiltonian that has the same numerical value,
  but possibly a different vertex set.
\end{proof}
Note that the equivalence above is broken by any term in the
Hamiltonian that is dependent on the vertex set, such as an external
magnetic field, or if the Hamiltonian cannot be put into the form in
Eq.~\eqref{eq:6}.
\begin{corollary}
  The partition function (defined in Eq.~\eqref{eq:2}) on any finite
  tree graph on $N$ vertices is identical to the partition function on
  a 1D chain with $N$ vertices for Hamiltonians of the form in
  Eq.~\eqref{eq:6}.
\end{corollary}
\begin{proof}
  This follows from Theorem \ref{thm:tree1d} and the fact that the
  grand partition functions sum over all vertex (spin) configurations.
\end{proof}
Note that the one-one mapping in the above construction ensures that
each of the $q^N$ configurations is counted only once, and the constraint that the
Hamiltonian depends only on the edge set ensures that a different
correspondence between energies (numerical value of the Hamiltonian)
and spin configuration in the tree, relative to the chain, does not
impact the result. We reiterate that this latter part is violated by
terms such as an external magnetic field, as considered in Section
\ref{sec:ising-model-mag}.
  
In the following, this equivalence will become central to some of the
results we derive using the linked cluster expansion. Nevertheless,
caution is warranted when it comes to tree graphs --- our results
above imply that the zero-field Ising partition function on the Cayley
tree (which has a well-defined infinite volume limit) must be identical to the 1D model. This is true, nevertheless,
the Cayley tree shows a finite temperature critical point as shown by
the zero-field susceptibility\cite{matsuda1974}.

\section{Counting clusters}
\label{sec:counting-clusters}

The linked-cluster expansion requires an enumeration of the number of
embeddings of a subgraph $H$ in a graph $G$, also known as a lattice
constant\footnote{Sykes \emph{et al}\cite{sykes66} make a distinction
  between a strong embedding and a weak embedding. A strong embedding
  implies that all edges between vertices present in the graph are
  present in the cluster, whereas a weak embedding does not require
  that. The LCE we use relies on weak embdeddings.}. Following Sykes
\emph{et al}\cite{sykes66}, the weight of a particular cluster (graph) $c$ in the
expansion is given by
\begin{equation}
  \label{eq:8}
  W_{c}(O) = O(c) - \sum_{s\in c} M_{s}W_{s}(O)
\end{equation}
where $s$ are all subgraphs that can be embedded in cluster $c$. In
this expression, $O$ corresponds to any extensive observable such as
the logarithm of the partition function, or quantities that can be
derived from it, and $M_{s}$ corresponds to the multiplicity of the
subgraph $s$ in the graph $c$. This quantity, also known as a lattice
constant, provides an enumeration of the number of ways in which $s$
can be embedded in $c$.  Equation \eqref{eq:8} then gives us an
iterative procedure where the weight of the smallest cluster is equal
to the value of the observable; other weights can be obtained
sequentially.  Once the weights are obtained, the value of the
observable per unit volume can be obtained via
\begin{equation}
  \label{eq:18}
  \lim_{N'\to\infty,N'\gg N}\frac{O}{N'} = \sum_c M'_c W_c,
\end{equation}
where $M'_c$ are multiplicities per unit volume of the infinite
system. In other words, the $M_c'$ enumerate how many ways the
cluster $c$ can be embedded in a much larger system ($N'\gg N$)
divided by the number of vertices (or a corresponding volume-like
quantity) of that larger system.

For the classical Ising model on a 1D chain (with no magnetic field),
one can show that $W_{j}=0$ for $j\geq 3$ and
$O=\log Z$\cite{iyer2015}\footnote{This results generalizes to the
  $q$-state Potts model as can be verified directly by computing
  $\log Z_N$ for $N=1,2,3,4,\ldots$ and calculating the
  corresponding weights. By obtaining a formula for $\log Z_N$, we can show
  that $W_j=0$ for $j\geq 3$, just as in the spin-$\nicefrac12$ classical Ising
  model}.  Given Theorem \ref{thm:tree1d}, it follows that this
statement is true on all trees for all quantities that can be derived
from the partition function, since the clusters used to calculate
weights are all finite clusters.

We are interested in calculating a generic linked cluster expansion
for a Bethe lattice.  In this, we use the distinction made by Baxter
between a Bethe lattice and a Cayley tree.  On a Cayley tree, the
number of boundary vertices (that are attached to only one edge) scale
with the number of vertices. For a tree of degree 3 with $N$ vertices,
the number of boundary points is $N/2+1$.  Therefore, the boundary
does not become irrelevant in the infinite volume limit. On the other
hand, the Bethe lattice does not have a boundary and cannot be thought
of as the ``bulk'' of a Cayley tree, since there is no consistent way
to terminate this ``bulk''; it is always infinitely large.
Nevertheless, we will see that in the linked cluster expansion we can
treat both the Bethe lattice, and the Cayley tree in the infinite
volume limit.

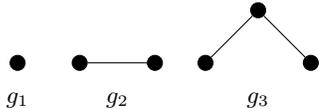
\begin{figure}[h]
  \centering
  \begin{tikzpicture}
    \node[shape=circle,draw=black,fill,inner sep=2pt] (a) at (0,0) {};
    \node (d) at (0,-0.5) {$g_{1}$};
  \end{tikzpicture}\hspace{10pt}
  \begin{tikzpicture}
    \node[shape=circle,draw=black,fill,inner sep=2pt] (a) at (1,0) {};
    \node[shape=circle,draw=black,fill,inner sep=2pt] (b) at (2,0) {};
    \path[-] (a) edge (b); \node (d) at (1.5,-0.5) {$g_{2}$};
  \end{tikzpicture}\hspace{10pt}
  \begin{tikzpicture}
    \node[shape=circle,draw=black,fill,inner sep=2pt] (a) at (0,0) {};
    \node[shape=circle,draw=black,fill,inner sep=2pt] (b) at (0.7,0.7)
    {}; \node[shape=circle,draw=black,fill,inner sep=2pt] (c) at
    (1.4,0) {}; \path[-] (a) edge (b); \path[-] (b) edge (c); \node
    (d) at (0.7,-0.5) {$g_{3}$};
  \end{tikzpicture}
  \caption{Some graphs and corresponding notation.}
  \label{fig:basicGraphs}
\end{figure}

The computation of the multiplicities of various clusters needs to be
carried out separately for the Cayley tree and the Bethe lattice. To
motivate this, we first consider the 1D chain.

In a 1D chain of $N$ vertices, there are always $N+1-j$ ways to embed
a $j$-chain, so that $M_j = N+1-j$. The linked cluster expansion then
becomes straightforward. From Eq.~\eqref{eq:8}, we get
\begin{equation}
  \label{eq:5}
  W_{N}(O) = O_{N} - \sum_{j=1}^{N-1} (N+1-j)W_{j}(O).
\end{equation}
The result for the observable per unit volume is then given by
Eq.~\eqref{eq:18} Here we consider the ``infinite'' system to have
size $N'\gg N$. Since $M'_c$ are the multiplicities of clusters
embedded in the infinite system per unit volume, we get
$M'_j = (N'+1-j)/N' \to 1$ since $j\leq N\ll N'$. Combining this with
Eq.~\eqref{eq:18} we get a particularly simple result for 1D chains,
\begin{equation}
  \label{eq:39}
  \lim_{N'\to\infty}O/N'= \sum_{j=1}^{N}W_j.
\end{equation}
It is critical to note that this result is applicable for all models
on a 1D-chain, quantum or classical, since the result does not assume
an underlying Hamiltonian or a specific observable $O$.

Consider now an infinite ($N'\to\infty$ vertices) 2D square lattice. Figure~\ref{fig:basicGraphs} shows our notation
for some graphs. For $c=g_{1}$, there are $N'$ ways
to embed this cluster, and we obtain $M'_{g_1} = 1$. For $c=g_{2}$,
there are $\sim 2N'$ ways to embed it because of the vertical and
horizontal edges, so we end up with $M'_{g_{2}} = 2$ in the limit. In
this fashion, the multiplicity has to be computed for every
cluster $c$.  Table \ref{tab:Mc-2D} shows the multiplicities for some
clusters embedded in infinite 2D square and triangular
lattices. Again, we note that these values are independent of the
Hamiltonian, or whether it is quantum or classical, and only depend on
the structure of the lattice.
\begin{table}[h]
  \centering
  \begin{tabular}{|c|c|c|}
    \hline
    Cluster & $M'_{c,\square}$ & $M'_{c,\triangle}$ \\
    \hline
    \begin{tikzpicture}
      \node[shape=circle,draw=black,fill,inner sep=1pt] (a) at (0,0)
      {};\end{tikzpicture}
            & 1 & 1\\
    \hline
    \begin{tikzpicture}
      \node[shape=circle,draw=black,fill,inner sep=1pt] (a) at
      (-0.2,0) {}; \node[shape=circle,draw=black,fill,inner sep=1pt]
      (b) at (0.2,0) {}; \path[-] (a) edge (b);
    \end{tikzpicture}
            & 2 & 3\\
    \hline
    \begin{tikzpicture}
      \node[shape=circle,draw=black,fill,inner sep=1pt] (a) at (0,0)
      {}; \node[shape=circle,draw=black,fill,inner sep=1pt] (b) at
      (0.2,0.2) {}; \node[shape=circle,draw=black,fill,inner sep=1pt]
      (c) at (0.4,0) {}; \path[-] (a) edge (b); \path[-] (b) edge (c);
    \end{tikzpicture} & 6  & 15 \\
    \hline
    \begin{tikzpicture}
      \node[shape=circle,draw=black,fill,inner sep=1pt] (a) at (0,0)
      {}; \node[shape=circle,draw=black,fill,inner sep=1pt] (b) at
      (0.2,0.2) {}; \node[shape=circle,draw=black,fill,inner sep=1pt]
      (c) at (0.4,0) {}; \path[-] (a) edge (b); \path[-] (b) edge (c);
      \path[-] (a) edge (c);
    \end{tikzpicture} & 0 & 2 \\
    \hline
    \begin{tikzpicture}
      \node[shape=circle,draw=black,fill,inner sep=1pt] (a) at (0,0)
      {}; \node[shape=circle,draw=black,fill,inner sep=1pt] (b) at
      (0.2,0.2) {}; \node[shape=circle,draw=black,fill,inner sep=1pt]
      (c) at (0.4,0) {}; \node[shape=circle,draw=black,fill,inner
      sep=1pt] (d) at (0.6,0.2) {}; \path[-] (a) edge (b); \path[-]
      (b) edge (c); \path[-] (c) edge (d);
    \end{tikzpicture} & 18 & 69\\
    \hline
    \begin{tikzpicture}
      \node[shape=circle,draw=black,fill,inner sep=1pt] (a) at (0,0)
      {}; \node[shape=circle,draw=black,fill,inner sep=1pt] (b) at
      (0.2,0.14) {}; \node[shape=circle,draw=black,fill,inner sep=1pt]
      (c) at (0.4,0) {}; \node[shape=circle,draw=black,fill,inner
      sep=1pt] (d) at (0.2,0.4) {}; \path[-] (a) edge (b); \path[-]
      (b) edge (c); \path[-] (b) edge (d);
    \end{tikzpicture} & 4 & 18 \\
    \hline
    \begin{tikzpicture}
      \node[shape=circle,draw=black,fill,inner sep=1pt] (a) at (0,0)
      {}; \node[shape=circle,draw=black,fill,inner sep=1pt] (b) at
      (0,0.3) {}; \node[shape=circle,draw=black,fill,inner sep=1pt]
      (c) at (0.3,0.3) {}; \node[shape=circle,draw=black,fill,inner
      sep=1pt] (d) at (0.3,0) {}; \path[-] (a) edge (b); \path[-] (b)
      edge (c); \path[-] (c) edge (d); \path[-] (d) edge (a);
    \end{tikzpicture} & 1 & 2\\
    \hline
    \begin{tikzpicture}
      \node[shape=circle,draw=black,fill,inner sep=1pt] (a) at (0,0)
      {}; \node[shape=circle,draw=black,fill,inner sep=1pt] (b) at
      (0,0.3) {}; \node[shape=circle,draw=black,fill,inner sep=1pt]
      (c) at (0.3,0.3) {}; \node[shape=circle,draw=black,fill,inner
      sep=1pt] (d) at (0.3,0) {}; \path[-] (a) edge (b); \path[-] (b)
      edge (c); \path[-] (c) edge (d); \path[-] (d) edge (a); \path[-]
      (a) edge (c);
    \end{tikzpicture} & 0 & 2 \\
    \hline
  \end{tabular}
  \caption{Combinatoric factors $M_{c}$ for various clusters in the 2D
    square and triangular lattices}
  \label{tab:Mc-2D}
\end{table}

We now proceed to obtain the multiplicities $M'_{c}$ for the Bethe
lattice and the Cayley tree. Before we proceed, we reiterate the
distinction between the two using the difference in the average degree
of a vertex. In a Cayley tree on $N$ vertices, there are $N/2+1$
boundary vertices that have degree 1, and all other vertices have
degree 3.  The average degree (also known as connectivity) is then
\begin{equation}
  \label{eq:12}
  c_{\mathrm{CT}} = \frac{1}{N}\left[3\left(\frac{N}{2}-1\right) + \left(\frac{N}{2}+1\right)\right] = 2 - \frac{2}{N}.
\end{equation}
For an $m$-Cayley tree, the bulk vertices have degree $m$ and the
boundary vertices have degree 1. For a graph with $N$ vertices, there
are $\frac{N(m-2)+2}{m-1}$ boundary vertices, and we recover the above
result for arbitrary $m$.

For a Bethe lattice, however, there is no boundary, and all vertices
have the same degree, $m$, giving
\begin{equation}
  \label{eq:13}
  c_{\rm{BL}} = m.
\end{equation}
In the thermodynamic limit of the Cayley tree, we approach $c=2$ for
any $m$, which is identical to the connectivity of a 1D chain. This
is one way of understanding why the Cayley tree has the same partition
function as the 1D chain. Another way is using the equivalence
established in Sec.~\ref{sec:Ztree}.

Below we develop the linked cluster expansion for these $m=3$ lattices
and obtain the partition function for the Cayley Tree and the Bethe
Lattice.

\subsection{Cayley tree}
\label{sec:cayley-tree}

First, we calculate multiplicities for the Cayley tree and show that
we indeed reproduce the result of a 1D lattice.

Consider a Cayley tree with $N'$ vertices.  We have
$M'_{g_{1}} = N'/N'\to1$.  For $g_{2}$, we count a total of $N'-1$
edges (each edge is a $g_{2}$) giving us $M'_{g_{2}}=(N'-1)/N'\to 1$
in the limit of large $N'$.  For $g_{3}$, each vertex has three ways
of embedding $g_{3}$, except the boundary vertices.  We therefore have
to subtract $N'/2+1$ vertices since one cannot embed a $g_{3}$
centered on a boundary vertex.  This gives us
$M'_{g_{3}} = 3(N'/2-1)/N'\to 3/2$. The calculation becomes more
tedious from here onwards due to increasing complexity of the
clusters. Table \ref{tab:Mc-Bethe} shows the multiplicities for some
higher order clusters.
\begin{table}[h]
  \centering
  \begin{tabular}{|c|c|c|}
    \hline
    Cluster & $M'_{c,CT}$ & $M'_{c,BL}$\\
    \hline
    \begin{tikzpicture}
      \node[shape=circle,draw=black,fill,inner sep=1pt] (a) at (0,0)
      {};\end{tikzpicture}
            & 1 & 1\\
    \hline
    \begin{tikzpicture}
      \node[shape=circle,draw=black,fill,inner sep=1pt] (a) at
      (-0.2,0) {}; \node[shape=circle,draw=black,fill,inner sep=1pt]
      (b) at (0.2,0) {}; \path[-] (a) edge (b);
    \end{tikzpicture}
            & 1 & $\nicefrac32$ \\
    \hline
    \begin{tikzpicture}
      \node[shape=circle,draw=black,fill,inner sep=1pt] (a) at (0,0)
      {}; \node[shape=circle,draw=black,fill,inner sep=1pt] (b) at
      (0.2,0.2) {}; \node[shape=circle,draw=black,fill,inner sep=1pt]
      (c) at (0.4,0) {}; \path[-] (a) edge (b); \path[-] (b) edge (c);
    \end{tikzpicture} & $\nicefrac32$ & 3 \\
    \hline
    \begin{tikzpicture}
      \node[shape=circle,draw=black,fill,inner sep=1pt] (a) at (0,0)
      {}; \node[shape=circle,draw=black,fill,inner sep=1pt] (b) at
      (0.2,0.2) {}; \node[shape=circle,draw=black,fill,inner sep=1pt]
      (c) at (0.4,0) {}; \node[shape=circle,draw=black,fill,inner
      sep=1pt] (d) at (0.6,0.2) {}; \path[-] (a) edge (b); \path[-]
      (b) edge (c); \path[-] (c) edge (d);
    \end{tikzpicture} & 2 & 6\\
    \hline
    \begin{tikzpicture}
      \node[shape=circle,draw=black,fill,inner sep=1pt] (a) at (0,0)
      {}; \node[shape=circle,draw=black,fill,inner sep=1pt] (b) at
      (0.2,0.14) {}; \node[shape=circle,draw=black,fill,inner sep=1pt]
      (c) at (0.4,0) {}; \node[shape=circle,draw=black,fill,inner
      sep=1pt] (d) at (0.2,0.4) {}; \path[-] (a) edge (b); \path[-]
      (b) edge (c); \path[-] (b) edge (d);
    \end{tikzpicture} & \nicefrac12 & 1\\
    \hline
    \begin{tikzpicture}
      \node[shape=circle,draw=black,fill,inner sep=1pt] (a) at (0,0)
      {}; \node[shape=circle,draw=black,fill,inner sep=1pt] (b) at
      (0.2,0.2) {}; \node[shape=circle,draw=black,fill,inner sep=1pt]
      (c) at (0.4,0) {}; \node[shape=circle,draw=black,inner sep=1pt]
      (d) at (0.6,0.2) {}; \node[shape=circle,draw=black,fill,,inner
      sep=1pt] (e) at (0.8,0)
      {};\node[shape=circle,draw=black,fill,inner sep=1pt] (f) at
      (1.0,0.2) {}; \node[shape=circle,draw=black,fill,inner sep=1pt]
      (g) at (1.2,0.0) {}; \path[-] (a) edge (b); \path[-] (b) edge
      (c); \path[densely dotted] (c) edge (d); \path[densely dotted]
      (d) edge (e); \path[-] (e) edge (f); \path[-] (f) edge (g);
      \draw[thin, decoration={ brace, mirror, raise=0.1cm }, decorate
      ] (a) -- (g) node [pos=0.5,anchor=north,yshift=-0.1cm] {$n$
        vertices};
    \end{tikzpicture} & $\begin{cases}
      2^{\frac{n}{2}-1},\, n \text{ even}, n\geq 2\\
      3\cdot 2^{\frac{n-5}{2}} ,\, n \text{ odd}, n\geq 3.
    \end{cases}$
            & $3\cdot 2^{n-3}, \, n\geq 1$\\
    \hline
  \end{tabular}
  \caption{Lattice constants (multiplicities) for various clusters on
    the Cayley tree ($M'_{c,CT}$) and the Bethe lattice ($M_{c,BL}'$)
    with $m=3$.}
  \label{tab:Mc-Bethe}
\end{table}
However, as noted before, the weights $W_{3}$ and beyond vanish for the classical spin-$\nicefrac12$ Ising model,
therefore these multiplicities are irrelevant, and we get a partition
function:
\begin{equation}
  \label{eq:14}
  \begin{split}
    -\beta f = \lim_{N'\to\infty}\frac{\log Z}{N'} &= M'_{1}W_{1} + M'_{2}W_{2} \\
    &=\log 2 + \log\cosh\beta J \\
    &= \log [2\cosh(\beta J)],
  \end{split}
\end{equation}
where $f$ is the free energy per unit volume (number of
vertices). This result is identical to the 1D chain. Nevertheless, we
note that the model has a known critical point that only becomes
manifest when we compute the zero-field susceptibility.

\subsection{Bethe lattice}
\label{sec:bethe-lattice}

The Bethe lattice is in a sense already in the thermodynamic limit
since it does not have a boundary. As seen from the average
connectivity, we cannot treat the Bethe lattice as the thermodynamic
limit of the Cayley tree. There is no consistent way to define the
partition function of a finite part of the Bethe lattice, since that
notion is ill-defined. The LCE gives us the partition function per
unit volume, so there is some hope that we can effectively divide by
the already infinite volume since we do not take a limit in the
process.

An application of the LCE gives us a different multiplicity
$M'_{g_{2}} = 3/2$ on the Bethe lattice; each vertex is connected to
three edges, and each edge is double counted. We cannot use the ``edge
counting'' method we used for the Cayley tree because one cannot
terminate the Bethe lattice. More generally, we cannot calculate the
total number of ways of embedding a given cluster in a ``finite
but large'' graph and then divide out by the volume and take the limit. A
``finite but large'' graph does not exist for the Bethe lattice. In other to
calculate multiplicities, we have to work ``intensively'' by counting
the number of ways to embed a given cluster at a given vertex and then
correcting for any multiple-counting. For the case of a lattice whose
boundary is always negligible in the limit (1D chain, square, etc.)
these two methods coincide.

Further, since $W_{3}$ and above are zero, we do not need to calculate
higher multiplicities (see Table~\ref{tab:Mc-Bethe} for some of these;
we use them in the susceptibility calculation in Section
\ref{sec:ising-model-mag}), and we end up with a partition function
per unit volume given by
\begin{equation}
  \label{eq:15}
  -\beta f = \log2 + \frac32\log\cosh(\beta J)
\end{equation}
where $f$ is the free energy per site.  This, remarkably, is the
correct free energy for the Ising model on the Bethe
lattice\cite{baxter2007}.  The result is an analytic function of
$\beta$ and therefore one might naively assume that there is no phase
transition for $\beta < \infty$. However, this is a known oddity with
the Ising model on the Bethe lattice, and the model indeed has a
finite temperature phase transition that only becomes manifest when
one computes the zero-field magnetic susceptibility.

We note generally that in an Ising model with a phase transition, at
$T<T_{c}$, in the absence of an external magnetic field, there is
nothing to break the symmetry to determine whether the majority of the
spins point up or down. On a finite lattice, one could pin boundary
spins, but we do not have that luxury on the Bethe lattice.  The only
option we're left with is to calculate the free energy in the presence
of a magnetic field, find the susceptibility, and study it for
non-analyticity.

\section{Ising model with a magnetic field}
\label{sec:ising-model-mag}

We now turn on a small magnetic field $H\ll J$ and study the free
energy in the presence of this small field:
\begin{equation}
  \label{eq:11}
  H = -\sum_{\langle i j \rangle}Js_{i}s_{j} - H\sum_{j}s_{j}
\end{equation}
For $J>0$, the model is ferromagnetic. At zero temperature, all spins
are aligned and point along the external field.  Note that in the
presence of a magnetic field, the conditions for Theorem
\ref{thm:tree1d} no longer hold and generally the tree is not
equivalent to a 1D chain.  In principle then, the partition function
on all branched clusters will have to be calculated separately, and
one does not generally expect their weights to go to
zero. Nevertheless, a simplification occurs at lowest order in the
external field.

Since the free energy has to be an even function of $H$, we will keep
terms to $O(H^{2})$ and discard the rest. We begin with the 1D chain.

\begin{figure}[h]
  \begin{tikzpicture}
    \node[shape=circle,draw=black,fill,inner sep=2pt] (a) at (0,0) {};
    \node[shape=circle,draw=black,fill,inner sep=2pt] (b) at (0.7,0.7)
    {}; \node[shape=circle,draw=black,fill,inner sep=2pt] (c) at
    (1.4,0) {}; \node[shape=circle,draw=black,fill,inner sep=2pt] (d)
    at (2.1,0.7) {}; \path[-] (a) edge (b); \path[-] (b) edge (c);
    \path[-] (c) edge (d); \node (e) at (1.05,-0.5) {$g_{4}$};
  \end{tikzpicture}\hspace{10pt}
  \begin{tikzpicture}
    \node[shape=circle,draw=black,fill,inner sep=2pt] (a) at (0,0) {};
    \node[shape=circle,draw=black,fill,inner sep=2pt] (b) at (0.7,0.5)
    {}; \node[shape=circle,draw=black,fill,inner sep=2pt] (c) at
    (1.4,0) {}; \node[shape=circle,draw=black,fill,inner sep=2pt] (d)
    at (0.7,1.25) {}; \path[-] (a) edge (b); \path[-] (b) edge (c);
    \path[-] (b) edge (d); \node (e) at (0.7,-0.5) {$g_{4'}$};
  \end{tikzpicture}
  \caption{4-vertex graphs}
  \label{fig:4vertex}
\end{figure}
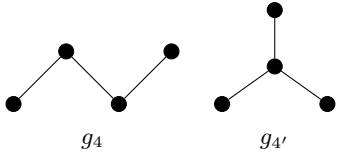

\subsection{1D chain}
\label{sec:1d-chain-mag}

We study the weights for graphs $g_{1}$ ,$g_{2}$, and $g_{3}$.
Denoting $\log Z_{g} = -\beta F_{g}$ for a graph $g$, and switching to
variables $K=\beta J$ and $h=\beta H$, we have
\begin{equation}
  \label{eq:19}
  \begin{split}
    -\beta F_{g_{1}} &= \log2 + \log\cosh( h)\approx \log2 + \frac{h^{2}}{2},\\
    -\beta F_{g_{2}} &\approx \log[4\cosh(K)] +h^{2}[1+\tanh(\beta J)],\\
    -\beta F_{g_{3}} &\approx \log[8\cosh^{2}(K)] + \\
    & \frac{h^{2}}{4}\cosh^{-2}(K)\big[1 + 5\cosh(2K) +
    4\sinh(2K)\big].
  \end{split}
\end{equation}
Calculating the weights, we get
\begin{equation}
  \label{eq:20}
  \begin{split}
    W_{g_{1}} &= \log2 + \frac{h^{2}}{2},\\
    W_{g_{2}} &= \log\cosh(K) + h^{2}\tanh(K),\\
    W_{g_{3}} &= h^{2}\tanh^{2}(K).
  \end{split}
\end{equation}
In this case, $W_{3}$ does not go to zero to $O(h^{2})$. This is to be
expected.  At the next level to the same order in $h$, we get
\begin{equation}
  \label{eq:21}
  W_{g_{4}} = h^{2}\tanh^{3}(K).
\end{equation}

Carrying on, at generic order $n\geq 3$, we get
\begin{equation}
  \label{eq:22}
  W_{g_{n}} = h^{2}\tanh^{n-1}(K),
\end{equation}
giving
\begin{equation}
  \label{eq:23}
  \begin{split}
    -\beta f_{\mathrm{1D}} &= \sum_{j=1}^{\infty}W_{j}\\
    &= \log[2\cosh(K)] + h^{2}\left[\frac12 +
      \sum_{j=1}^{\infty}\tanh^{j}(K)\right]\\
    &= \log[2\cosh(K)] + \frac{h^{2}}{2}e^{2K},
  \end{split}
\end{equation}
where we can sum the series for all $0\leq K < \infty$, since
$\tanh(K)<1$ in this range.  We can obtain the low-field magnetization
density by
\begin{equation}
  \label{eq:24}
  m_{\mathrm{1D}} = -\beta\frac{\partial f_{\mathrm{1D}}}{\partial h} = h e^{2K} =  \frac{H}{k_{B}T} e^{2J/k_{B}T}.
\end{equation}
There are clearly no singularities in $T$ at $h=0^{+}$, consistent
with the fact that there is no ordered phase for $T>0$ and zero field
for the 1D chain.


\subsection{Cayley tree}
\label{sec:cayley-tree-1}

Considering only the chain like clusters (a direct calculation of
branched clusters at this order shows that they do not contribute to
the weights; see Appendix \ref{appsec:branched}), we use the weights
derived above and calculate the free energy density from the
multiplicities in Table~\ref{tab:Mc-Bethe} to get
\begin{multline}
  \label{eq:29}
  -\beta f_{CT} = \log[2\cosh(K)] + \frac{h^{2}}{2}\Bigg[1 + \\
  + \left\{\frac{1}{\tanh(K)} +
    \frac32\right\}\sum_{j=1}^{\infty}2^{j}\tanh^{2j}(K)\Bigg].
\end{multline}
The above sum converges for $2\tanh^{2}(K)<1$. We get
\begin{equation}
  \label{eq:30}
  -\beta f_{CT} = \log[2\cosh(K)] + \frac{h^{2}}{2}\frac{\{1+\tanh(K)\}^{2}}{1-2\tanh^{2}(K)}.
\end{equation}
The above expression has a singularity at
$K_{c,CT}=\tanh^{-1}(1/\sqrt{2})$ indicating a critical point. This is
in fact a well-known result, and what we see here is only the first of
a chain of critical points from $K_{c,CT}$ to $K_{c,BL}$ obtained in
the next section\cite{muller-hartmann1974}. The other critical points
appear at higher order in $h$ and for $K<K_{c,CT}$. However, at higher
order, the branched clusters cannot be neglected and it is not
straightforward to obtain the other singularities analytically using
this method.

\subsection{Bethe lattice}
\label{sec:bethe-lattice-mag}

First, we note that the multiplicity for $g_{3}$ on the Bethe lattice
is given by $M'_{g_{3}} = 3$ since we can embed $g_{3}$ in 3 ways at
every vertex.  For $g_{4}$, we note that starting at any vertex of the
Bethe lattice, we can choose a ``path'' for $g_{4}$ in $3\times2\times2$
ways. Since the opposite path exists starting at a different vertex,
this path is double counted, giving $M'_{g_{4}}= 6$.  In fact, this can
be immediately generalized to all the chain graphs,
$M'_{g_{n}} = 3\times 2^{n-3}$ for $n\geq 2$.  The graphs with branches
are a little more complicated. $g_{4'}$ embeds uniquely at every
vertex, and therefore, $M'_{g_{4'}}=1$. Note that for the chain graphs,
the multiplicities rise exponentially.  For the branched graphs,
however, each time a new branch is introduced the multiplicity falls
because the Bethe lattice has a very specific branching structure.
For a fixed branch structure, the multiplicities grow as we make the
chain longer; each branched structure then produces its own cascade of
chains.

We begin by only considering the chain graphs since, like for the
Cayley Tree, the branched clusters do not contribute at this order in
$h$. This leads to a free energy given by
\begin{multline}
  \label{eq:26}
  -\beta f_{\mathrm{BL}} =\log2 + \frac32\log\cosh(K) +\\
  + h^{2}\left[\frac12 + \frac34\sum_{j=1}^{\infty}(2\tanh
    K)^{n}\right].
\end{multline}
The sum in the equation above can only be carried out for
$0\leq K < \tanh^{-1}(1/2)$ signaling the possibility of a finite
temperature phase transition. For $K$ in this range, we get
\begin{equation}
  \label{eq:25}
  -\beta f_{\mathrm{BL}} =\log2 + \frac32\log\cosh(K) +
  \frac{h^{2}}{2}\frac{1+\tanh(K)}{1-2\tanh(K)}.
\end{equation}
The corresponding zero-field susceptibility is given by
\begin{equation}
  \label{eq:10}
  \chi_{\mathrm{BL}} = \beta\frac{1+\tanh(K)}{1-2\tanh(K)}.
\end{equation}
\begin{figure}[h]
  \centering \includegraphics[width=8cm]{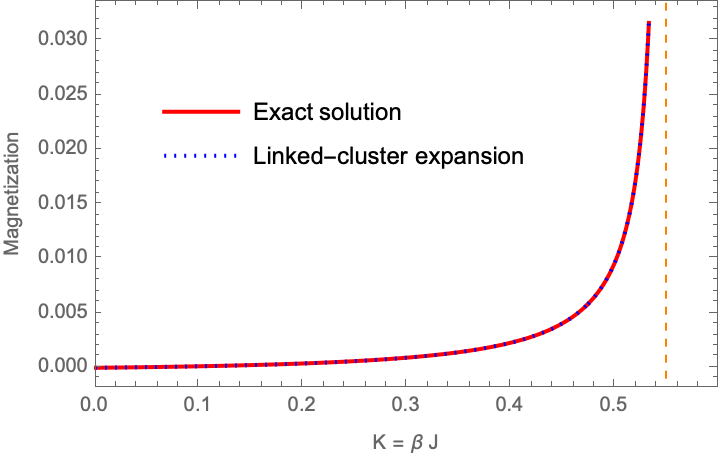}
  \caption{Magnetization as a function of inverse temperature
    $K = \beta J$ in the disordered phase upto the critical point for
    $m=3$. Note that the two methods indicated give exactly the same
    results. In this plot, $h/K = 0.001$. As $h/K\to0$, we are able to
    access regions closer to the critical point. The orange dotted
    line marks the critical point. The equality between the two
    methods holds for $m>3$.}
  \label{fig:magnetization}
\end{figure}
As indicated above, the 0-field susceptibility has a singularity as
$K_{c} = \beta_{c} J \to \tanh^{-1}(1/2)$ from below, precisely the
temperature of the known phase transition of the Ising model on the
Bethe lattice\cite{baxter2007}. For $T \to\infty $, we retrieve the
result of the 1D chain as one can check from a high temperature
expansion, since at infinite temperature, the spins are all
uncorrelated and the magnetization per unit volume is proportional to
the external field and inversely proportional to the temperature.  In
fact, we will see that we recover the correct low-field magnetization
for $T>T_{c}$, i.e., the unordered phase (see
Fig.~\ref{fig:magnetization}). Our calculation thus far cannot predict
the free energy for the ordered phase or the correct magnetization
discontinuity at the critical point. Indeed, our formula gives us the
unphysical result that the magnetization density goes to infinity at
the critical point.  Nevertheless, we are able to extract the critical
temperature.

For an $m$-Bethe lattice, the chain graphs $g_{n}$ for $n\geq 2$ have
multiplicities given by $M'_{g_{n}} = m(m-1)^{n-2}/2$. Considering
again only the chain-like graphs (the branched clusters do not
contribute at this order for all $m$), we can find the corresponding
free energy
\begin{multline}
  \label{eq:17}
  -\beta f_{m\mathrm{BL}} = \log2 + \frac{m}{2}\log\cosh(K) +  \\
  \frac{h^{2}}{2}\left[ \frac{1+\tanh(K)}{1-(m-1)\tanh(K)} \right].
\end{multline}
We therefore obtain a critical temperature of
\begin{equation}
  \label{eq:16}
  \beta_{c}J = \tanh^{-1}\left(\frac{1}{m-1}\right) = \frac12 \log\left(\frac{m}{m-2}\right).
\end{equation}
where the second equality is the more familiar form of this
expression.

The linked-cluster expansion therefore not only gives us the critical
point, but also does much better than a high temperature expansion,
giving us the low-field result in the entire disordered phase for
arbitrary $m$ with minimal effort. Methods using self-similarity (the
general approach to the exact solution) can give us the same results
(see e.g.~Refs.\cite{matsuda1974},\cite{muller-hartmann1974}).
However, the LCE shows us what clusters contribute to the critical
point.  We reiterate that we are able to produce the results of a
mean-field Bethe-Peierls approximation for the disordered phase close
to the phase transition at small fields without engaging with the
self-similar nature of the lattice (equivalent to the consistency
conditions imposed in the Bethe approximation). In fact, we only
consider all possible linear chains, and our result derives purely
from the combinatorics of placing these linear chains in $m$-Bethe
lattices. The various branched clusters do not appear to contribute to
the critical temperature, since it is a zero-field property.

\section{Conclusions}
\label{sec:conclusions}

In this work, we extend the linked cluster expansion to tree graphs,
and obtain lattice constants and multiplicities of the various cluster
embeddings. In particular, we show how the subtle difference between
the Cayley tree and Bethe lattice can be captured using this method,
leading to different results for the corresponding partition
functions. The derived lattice constants can be applied to any
classical or quantum model on these trees, since they depend only on
the lattice structure and not the Hamiltonians or the observables
calculated, therefore laying the groundwork for several future studies
that will potentially extend the precision of numerical
approximations.

The Ising model on the Cayley tree and Bethe lattice have been
extensively studied (despite some confusion in the literature about
how these lattices are defined, especially the language surrounding
thermodynamic limits), and their critical properties are well
known. Nevertheless, in this work, our use of the linked cluster
expansion has revealed several interesting insights about these
models.

First, we have shown that for a classical spin Hamiltonian that can be
cast in the form of Eq.~\eqref{eq:6}, any finite tree lattice can be
mapped onto an equivalent 1D lattice for the purpose of computing the
partition function, and all properties derived from it. This result
explains why branched clusters and cluster with more than 3 vertices
do not contribute to the LCE for trees in the absence of an external
magnetic field.  We use this to show that the finite $m$-Cayley tree
in the absence of a external magnetic field is similar to a 1D chain
on the same number of vertices, and has the same free
energy. Nevertheless, in the thermodynamic limit, even though the free
energy stays identical to that of the 1D chain, the model develops a
nontrivial singularity in the zero-field susceptibility that can be
analytically obtained using the LCE, showing a departure from the 1D
chain.  We then show that the same method applied to the Bethe lattice
provides a different free energy solely from the different
combinatorics of embedding clusters.

We see that the linked cluster expansion, despite being a ``series''
type of approximation method that progressively counts larger finite
clusters, is capable of providing correct solutions for models where
neither a finite lattice, or a thermodynamic limit are well-defined,
indicating that it achieves an elimination of boundary contributions
at every stage or order of the expansion. This feature of the LCE is
the reasons why it overcomes the large [$O(1/N)$] errors that occur in 
calculations using a simple finite size extrapolation based on
the grand-canonical partition function in open systems, and why
we get more rapid convergence to the thermodynamic limit.

We reiterate that we have shown a novel way that the linked cluster
expansion can be used, in places where traditional finite size
extrapolations are fundamentally inapplicable.  We foresee
straightforward application of the developments in this work to
quantum models, including disordered models, and periodically driven
models, thus allowing an alternate method to studying complex
phenomena that can be modeled using trees.

\bibliography{LCE-bethe.bib}

\begin{thebibliography}{26}%
\makeatletter
\providecommand \@ifxundefined [1]{%
 \@ifx{#1\undefined}
}%
\providecommand \@ifnum [1]{%
 \ifnum #1\expandafter \@firstoftwo
 \else \expandafter \@secondoftwo
 \fi
}%
\providecommand \@ifx [1]{%
 \ifx #1\expandafter \@firstoftwo
 \else \expandafter \@secondoftwo
 \fi
}%
\providecommand \natexlab [1]{#1}%
\providecommand \enquote  [1]{``#1''}%
\providecommand \bibnamefont  [1]{#1}%
\providecommand \bibfnamefont [1]{#1}%
\providecommand \citenamefont [1]{#1}%
\providecommand \href@noop [0]{\@secondoftwo}%
\providecommand \href [0]{\begingroup \@sanitize@url \@href}%
\providecommand \@href[1]{\@@startlink{#1}\@@href}%
\providecommand \@@href[1]{\endgroup#1\@@endlink}%
\providecommand \@sanitize@url [0]{\catcode `\\12\catcode `\$12\catcode
  `\&12\catcode `\#12\catcode `\^12\catcode `\_12\catcode `\%12\relax}%
\providecommand \@@startlink[1]{}%
\providecommand \@@endlink[0]{}%
\providecommand \url  [0]{\begingroup\@sanitize@url \@url }%
\providecommand \@url [1]{\endgroup\@href {#1}{\urlprefix }}%
\providecommand \urlprefix  [0]{URL }%
\providecommand \Eprint [0]{\href }%
\providecommand \doibase [0]{http://dx.doi.org/}%
\providecommand \selectlanguage [0]{\@gobble}%
\providecommand \bibinfo  [0]{\@secondoftwo}%
\providecommand \bibfield  [0]{\@secondoftwo}%
\providecommand \translation [1]{[#1]}%
\providecommand \BibitemOpen [0]{}%
\providecommand \bibitemStop [0]{}%
\providecommand \bibitemNoStop [0]{.\EOS\space}%
\providecommand \EOS [0]{\spacefactor3000\relax}%
\providecommand \BibitemShut  [1]{\csname bibitem#1\endcsname}%
\let\auto@bib@innerbib\@empty
\bibitem [{\citenamefont {Baxter}(2007)}]{baxter2007}%
  \BibitemOpen
  \bibfield  {author} {\bibinfo {author} {\bibfnamefont {R.J.}\ \bibnamefont
  {Baxter}},\ }\href {https://books.google.com/books?id=G3owDULfBuEC} {\emph
  {\bibinfo {title} {Exactly Solved Models in Statistical Mechanics}}},\ Dover
  books on physics\ (\bibinfo  {publisher} {Dover Publications},\ \bibinfo
  {year} {2007})\BibitemShut {NoStop}%
\bibitem [{\citenamefont {Thacker}(1981)}]{thacker81}%
  \BibitemOpen
  \bibfield  {author} {\bibinfo {author} {\bibfnamefont {H.~B.}\ \bibnamefont
  {Thacker}},\ }\bibfield  {title} {\enquote {\bibinfo {title} {Exact
  integrability in quantum field theory and statistical systems},}\ }\href
  {\doibase 10.1103/RevModPhys.53.253} {\bibfield  {journal} {\bibinfo
  {journal} {Rev. Mod. Phys.}\ }\textbf {\bibinfo {volume} {53}},\ \bibinfo
  {pages} {253--285} (\bibinfo {year} {1981})}\BibitemShut {NoStop}%
\bibitem [{\citenamefont {Korepin}\ \emph {et~al.}(1993)\citenamefont
  {Korepin}, \citenamefont {Bogoliubov},\ and\ \citenamefont
  {Izergin}}]{korepin1993}%
  \BibitemOpen
  \bibfield  {author} {\bibinfo {author} {\bibfnamefont {V.~E.}\ \bibnamefont
  {Korepin}}, \bibinfo {author} {\bibfnamefont {N.~M.}\ \bibnamefont
  {Bogoliubov}}, \ and\ \bibinfo {author} {\bibfnamefont {A.~G.}\ \bibnamefont
  {Izergin}},\ }\href {\doibase 10.1017/CBO9780511628832} {\emph {\bibinfo
  {title} {Quantum Inverse Scattering Method and Correlation Functions}}},\
  Cambridge Monographs on Mathematical Physics\ (\bibinfo  {publisher}
  {Cambridge University Press},\ \bibinfo {year} {1993})\BibitemShut {NoStop}%
\bibitem [{\citenamefont {Oitmaa}\ \emph {et~al.}(2006)\citenamefont {Oitmaa},
  \citenamefont {Hamer},\ and\ \citenamefont {Zheng}}]{oitmaa2006}%
  \BibitemOpen
  \bibfield  {author} {\bibinfo {author} {\bibfnamefont {Jaan}\ \bibnamefont
  {Oitmaa}}, \bibinfo {author} {\bibfnamefont {Chris}\ \bibnamefont {Hamer}}, \
  and\ \bibinfo {author} {\bibfnamefont {Weihong}\ \bibnamefont {Zheng}},\
  }\href {\doibase 10.1017/CBO9780511584398} {\emph {\bibinfo {title} {Series
  Expansion Methods for Strongly Interacting Lattice Models}}}\ (\bibinfo
  {publisher} {Cambridge University Press},\ \bibinfo {year}
  {2006})\BibitemShut {NoStop}%
\bibitem [{\citenamefont {Mitra}(2018)}]{mitra2018}%
  \BibitemOpen
  \bibfield  {author} {\bibinfo {author} {\bibfnamefont {Aditi}\ \bibnamefont
  {Mitra}},\ }\bibfield  {title} {\enquote {\bibinfo {title} {Quantum quench
  dynamics},}\ }\href {\doibase 10.1146/annurev-conmatphys-031016-025451}
  {\bibfield  {journal} {\bibinfo  {journal} {Annual Review of Condensed Matter
  Physics}\ }\textbf {\bibinfo {volume} {9}},\ \bibinfo {pages} {245--259}
  (\bibinfo {year} {2018})}\BibitemShut {NoStop}%
\bibitem [{\citenamefont {Iyer}\ \emph {et~al.}(2015)\citenamefont {Iyer},
  \citenamefont {Srednicki},\ and\ \citenamefont {Rigol}}]{iyer2015}%
  \BibitemOpen
  \bibfield  {author} {\bibinfo {author} {\bibfnamefont {Deepak}\ \bibnamefont
  {Iyer}}, \bibinfo {author} {\bibfnamefont {Mark}\ \bibnamefont {Srednicki}},
  \ and\ \bibinfo {author} {\bibfnamefont {Marcos}\ \bibnamefont {Rigol}},\
  }\bibfield  {title} {\enquote {\bibinfo {title} {Optimization of finite-size
  errors in finite-temperature calculations of unordered phases},}\ }\href
  {\doibase 10.1103/PhysRevE.91.062142} {\bibfield  {journal} {\bibinfo
  {journal} {Phys. Rev. E}\ }\textbf {\bibinfo {volume} {91}},\ \bibinfo
  {pages} {062142} (\bibinfo {year} {2015})}\BibitemShut {NoStop}%
\bibitem [{\citenamefont {Iyer}\ \emph {et~al.}(2017)\citenamefont {Iyer},
  \citenamefont {Srednicki},\ and\ \citenamefont {Rigol}}]{iyer2015err}%
  \BibitemOpen
  \bibfield  {author} {\bibinfo {author} {\bibfnamefont {Deepak}\ \bibnamefont
  {Iyer}}, \bibinfo {author} {\bibfnamefont {Mark}\ \bibnamefont {Srednicki}},
  \ and\ \bibinfo {author} {\bibfnamefont {Marcos}\ \bibnamefont {Rigol}},\
  }\bibfield  {title} {\enquote {\bibinfo {title} {Erratum: Optimization of
  finite-size errors in finite-temperature calculations of unordered phases
  [phys. rev. e 91, 062142 (2015)]},}\ }\href {\doibase
  10.1103/PhysRevE.96.039903} {\bibfield  {journal} {\bibinfo  {journal} {Phys.
  Rev. E}\ }\textbf {\bibinfo {volume} {96}},\ \bibinfo {pages} {039903}
  (\bibinfo {year} {2017})}\BibitemShut {NoStop}%
\bibitem [{\citenamefont {{Sykes}}\ \emph {et~al.}(1966)\citenamefont
  {{Sykes}}, \citenamefont {{Essam}}, \citenamefont {{Heap}},\ and\
  \citenamefont {{Hiley}}}]{sykes66}%
  \BibitemOpen
  \bibfield  {author} {\bibinfo {author} {\bibfnamefont {M.~F.}\ \bibnamefont
  {{Sykes}}}, \bibinfo {author} {\bibfnamefont {J.~W.}\ \bibnamefont
  {{Essam}}}, \bibinfo {author} {\bibfnamefont {B.~R.}\ \bibnamefont {{Heap}}},
  \ and\ \bibinfo {author} {\bibfnamefont {B.~J.}\ \bibnamefont {{Hiley}}},\
  }\bibfield  {title} {\enquote {\bibinfo {title} {{Lattice Constant Systems
  and Graph Theory}},}\ }\href {\doibase 10.1063/1.1705066} {\bibfield
  {journal} {\bibinfo  {journal} {Journal of Mathematical Physics}\ }\textbf
  {\bibinfo {volume} {7}},\ \bibinfo {pages} {1557--1572} (\bibinfo {year}
  {1966})}\BibitemShut {NoStop}%
\bibitem [{\citenamefont {Domb}(1960{\natexlab{a}})}]{domb60a}%
  \BibitemOpen
  \bibfield  {author} {\bibinfo {author} {\bibfnamefont {C.}~\bibnamefont
  {Domb}},\ }\bibfield  {title} {\enquote {\bibinfo {title} {On the theory of
  cooperative phenomena in crystals},}\ }\href {\doibase
  10.1080/00018736000101189} {\bibfield  {journal} {\bibinfo  {journal}
  {Advances in Physics}\ }\textbf {\bibinfo {volume} {9}},\ \bibinfo {pages}
  {149--244} (\bibinfo {year} {1960}{\natexlab{a}})}\BibitemShut {NoStop}%
\bibitem [{\citenamefont {Domb}(1960{\natexlab{b}})}]{domb60b}%
  \BibitemOpen
  \bibfield  {author} {\bibinfo {author} {\bibfnamefont {C.}~\bibnamefont
  {Domb}},\ }\bibfield  {title} {\enquote {\bibinfo {title} {On the theory of
  cooperative phenomena in crystals},}\ }\href {\doibase
  10.1080/00018736000101199} {\bibfield  {journal} {\bibinfo  {journal}
  {Advances in Physics}\ }\textbf {\bibinfo {volume} {9}},\ \bibinfo {pages}
  {245--361} (\bibinfo {year} {1960}{\natexlab{b}})}\BibitemShut {NoStop}%
\bibitem [{\citenamefont {Tang}\ \emph {et~al.}(2013)\citenamefont {Tang},
  \citenamefont {Khatami},\ and\ \citenamefont {Rigol}}]{tang_khatami_13}%
  \BibitemOpen
  \bibfield  {author} {\bibinfo {author} {\bibfnamefont {B.}~\bibnamefont
  {Tang}}, \bibinfo {author} {\bibfnamefont {E.}~\bibnamefont {Khatami}}, \
  and\ \bibinfo {author} {\bibfnamefont {M.}~\bibnamefont {Rigol}},\ }\bibfield
   {title} {\enquote {\bibinfo {title} {A short introduction to numerical
  linked-cluster expansions},}\ }\href@noop {} {\bibfield  {journal} {\bibinfo
  {journal} {Comput. Phys. Commun.}\ }\textbf {\bibinfo {volume} {184}},\
  \bibinfo {pages} {557} (\bibinfo {year} {2013})}\BibitemShut {NoStop}%
\bibitem [{\citenamefont {Rigol}(2014)}]{rigol2014}%
  \BibitemOpen
  \bibfield  {author} {\bibinfo {author} {\bibfnamefont {M.}~\bibnamefont
  {Rigol}},\ }\bibfield  {title} {\enquote {\bibinfo {title} {Quantum quenches
  in the thermodynamic limit},}\ }\href {\doibase
  10.1103/PhysRevLett.112.170601} {\bibfield  {journal} {\bibinfo  {journal}
  {Phys. Rev. Lett.}\ }\textbf {\bibinfo {volume} {112}},\ \bibinfo {pages}
  {170601} (\bibinfo {year} {2014})}\BibitemShut {NoStop}%
\bibitem [{\citenamefont {Tang}\ \emph
  {et~al.}(2015{\natexlab{a}})\citenamefont {Tang}, \citenamefont {Iyer},\ and\
  \citenamefont {Rigol}}]{tang2015a}%
  \BibitemOpen
  \bibfield  {author} {\bibinfo {author} {\bibfnamefont {Baoming}\ \bibnamefont
  {Tang}}, \bibinfo {author} {\bibfnamefont {Deepak}\ \bibnamefont {Iyer}}, \
  and\ \bibinfo {author} {\bibfnamefont {Marcos}\ \bibnamefont {Rigol}},\
  }\bibfield  {title} {\enquote {\bibinfo {title} {Thermodynamics of
  two-dimensional spin models with bimodal random-bond disorder},}\ }\href
  {\doibase 10.1103/PhysRevB.91.174413} {\bibfield  {journal} {\bibinfo
  {journal} {Phys. Rev. B}\ }\textbf {\bibinfo {volume} {91}},\ \bibinfo
  {pages} {174413} (\bibinfo {year} {2015}{\natexlab{a}})}\BibitemShut
  {NoStop}%
\bibitem [{\citenamefont {Tang}\ \emph
  {et~al.}(2015{\natexlab{b}})\citenamefont {Tang}, \citenamefont {Iyer},\ and\
  \citenamefont {Rigol}}]{tang2015b}%
  \BibitemOpen
  \bibfield  {author} {\bibinfo {author} {\bibfnamefont {Baoming}\ \bibnamefont
  {Tang}}, \bibinfo {author} {\bibfnamefont {Deepak}\ \bibnamefont {Iyer}}, \
  and\ \bibinfo {author} {\bibfnamefont {Marcos}\ \bibnamefont {Rigol}},\
  }\bibfield  {title} {\enquote {\bibinfo {title} {Quantum quenches and
  many-body localization in the thermodynamic limit},}\ }\href {\doibase
  10.1103/PhysRevB.91.161109} {\bibfield  {journal} {\bibinfo  {journal} {Phys.
  Rev. B}\ }\textbf {\bibinfo {volume} {91}},\ \bibinfo {pages} {161109}
  (\bibinfo {year} {2015}{\natexlab{b}})}\BibitemShut {NoStop}%
\bibitem [{\citenamefont {Mallayya}\ and\ \citenamefont
  {Rigol}(2017)}]{mallayya2017}%
  \BibitemOpen
  \bibfield  {author} {\bibinfo {author} {\bibfnamefont {Krishnanand}\
  \bibnamefont {Mallayya}}\ and\ \bibinfo {author} {\bibfnamefont {Marcos}\
  \bibnamefont {Rigol}},\ }\bibfield  {title} {\enquote {\bibinfo {title}
  {Numerical linked cluster expansions for quantum quenches in one-dimensional
  lattices},}\ }\href {\doibase 10.1103/PhysRevE.95.033302} {\bibfield
  {journal} {\bibinfo  {journal} {Phys. Rev. E}\ }\textbf {\bibinfo {volume}
  {95}},\ \bibinfo {pages} {033302} (\bibinfo {year} {2017})}\BibitemShut
  {NoStop}%
\bibitem [{\citenamefont {Mallayya}\ and\ \citenamefont
  {Rigol}(2021)}]{mallayya2021}%
  \BibitemOpen
  \bibfield  {author} {\bibinfo {author} {\bibfnamefont {Krishnanand}\
  \bibnamefont {Mallayya}}\ and\ \bibinfo {author} {\bibfnamefont {Marcos}\
  \bibnamefont {Rigol}},\ }\bibfield  {title} {\enquote {\bibinfo {title}
  {Prethermalization, thermalization, and fermi's golden rule in quantum
  many-body systems},}\ }\href {\doibase 10.1103/PhysRevB.104.184302}
  {\bibfield  {journal} {\bibinfo  {journal} {Phys. Rev. B}\ }\textbf {\bibinfo
  {volume} {104}},\ \bibinfo {pages} {184302} (\bibinfo {year}
  {2021})}\BibitemShut {NoStop}%
\bibitem [{\citenamefont {Gan}\ and\ \citenamefont {Hazzard}(2020)}]{gan2020}%
  \BibitemOpen
  \bibfield  {author} {\bibinfo {author} {\bibfnamefont {Johann}\ \bibnamefont
  {Gan}}\ and\ \bibinfo {author} {\bibfnamefont {Kaden R.~A.}\ \bibnamefont
  {Hazzard}},\ }\bibfield  {title} {\enquote {\bibinfo {title} {Numerical
  linked cluster expansions for inhomogeneous systems},}\ }\href {\doibase
  10.1103/PhysRevA.102.013318} {\bibfield  {journal} {\bibinfo  {journal}
  {Phys. Rev. A}\ }\textbf {\bibinfo {volume} {102}},\ \bibinfo {pages}
  {013318} (\bibinfo {year} {2020})}\BibitemShut {NoStop}%
\bibitem [{\citenamefont {Zhang}\ \emph {et~al.}(2016)\citenamefont {Zhang},
  \citenamefont {Khemani},\ and\ \citenamefont {Huse}}]{zhang2016}%
  \BibitemOpen
  \bibfield  {author} {\bibinfo {author} {\bibfnamefont {Liangsheng}\
  \bibnamefont {Zhang}}, \bibinfo {author} {\bibfnamefont {Vedika}\
  \bibnamefont {Khemani}}, \ and\ \bibinfo {author} {\bibfnamefont {David~A.}\
  \bibnamefont {Huse}},\ }\bibfield  {title} {\enquote {\bibinfo {title} {A
  floquet model for the many-body localization transition},}\ }\href {\doibase
  10.1103/PhysRevB.94.224202} {\bibfield  {journal} {\bibinfo  {journal} {Phys.
  Rev. B}\ }\textbf {\bibinfo {volume} {94}},\ \bibinfo {pages} {224202}
  (\bibinfo {year} {2016})}\BibitemShut {NoStop}%
\bibitem [{\citenamefont {Abou-Chacra}\ \emph {et~al.}(1973)\citenamefont
  {Abou-Chacra}, \citenamefont {Thouless},\ and\ \citenamefont
  {Anderson}}]{chacra1973}%
  \BibitemOpen
  \bibfield  {author} {\bibinfo {author} {\bibfnamefont {R}~\bibnamefont
  {Abou-Chacra}}, \bibinfo {author} {\bibfnamefont {D~J}\ \bibnamefont
  {Thouless}}, \ and\ \bibinfo {author} {\bibfnamefont {P~W}\ \bibnamefont
  {Anderson}},\ }\bibfield  {title} {\enquote {\bibinfo {title} {A self
  consistent theory of localization},}\ }\href {\doibase
  10.1088/0022-3719/6/10/009} {\bibfield  {journal} {\bibinfo  {journal}
  {Journal of Physics C: Solid State Physics}\ }\textbf {\bibinfo {volume}
  {6}},\ \bibinfo {pages} {1734--1752} (\bibinfo {year} {1973})}\BibitemShut
  {NoStop}%
\bibitem [{\citenamefont {Basko}\ \emph {et~al.}(2006)\citenamefont {Basko},
  \citenamefont {Aleiner},\ and\ \citenamefont {Altshuler}}]{basko2006}%
  \BibitemOpen
  \bibfield  {author} {\bibinfo {author} {\bibfnamefont {D.M.}\ \bibnamefont
  {Basko}}, \bibinfo {author} {\bibfnamefont {I.L.}\ \bibnamefont {Aleiner}}, \
  and\ \bibinfo {author} {\bibfnamefont {B.L.}\ \bibnamefont {Altshuler}},\
  }\bibfield  {title} {\enquote {\bibinfo {title} {Metal-insulator transition
  in a weakly interacting many-electron system with localized single-particle
  states},}\ }\href {\doibase https://doi.org/10.1016/j.aop.2005.11.014}
  {\bibfield  {journal} {\bibinfo  {journal} {Annals of Physics}\ }\textbf
  {\bibinfo {volume} {321}},\ \bibinfo {pages} {1126--1205} (\bibinfo {year}
  {2006})}\BibitemShut {NoStop}%
\bibitem [{\citenamefont {Savitz}\ \emph {et~al.}(2019)\citenamefont {Savitz},
  \citenamefont {Peng},\ and\ \citenamefont {Refael}}]{savitz2019}%
  \BibitemOpen
  \bibfield  {author} {\bibinfo {author} {\bibfnamefont {Samuel}\ \bibnamefont
  {Savitz}}, \bibinfo {author} {\bibfnamefont {Changnan}\ \bibnamefont {Peng}},
  \ and\ \bibinfo {author} {\bibfnamefont {Gil}\ \bibnamefont {Refael}},\
  }\bibfield  {title} {\enquote {\bibinfo {title} {Anderson localization on the
  bethe lattice using cages and the wegner flow},}\ }\href {\doibase
  10.1103/PhysRevB.100.094201} {\bibfield  {journal} {\bibinfo  {journal}
  {Phys. Rev. B}\ }\textbf {\bibinfo {volume} {100}},\ \bibinfo {pages}
  {094201} (\bibinfo {year} {2019})}\BibitemShut {NoStop}%
\bibitem [{\citenamefont {Ostilli}(2012)}]{ostilli2012}%
  \BibitemOpen
  \bibfield  {author} {\bibinfo {author} {\bibfnamefont {M.}~\bibnamefont
  {Ostilli}},\ }\bibfield  {title} {\enquote {\bibinfo {title} {Cayley trees
  and bethe lattices: A concise analysis for mathematicians and physicists},}\
  }\href {\doibase https://doi.org/10.1016/j.physa.2012.01.038} {\bibfield
  {journal} {\bibinfo  {journal} {Physica A: Statistical Mechanics and its
  Applications}\ }\textbf {\bibinfo {volume} {391}},\ \bibinfo {pages}
  {3417--3423} (\bibinfo {year} {2012})}\BibitemShut {NoStop}%
\bibitem [{\citenamefont {Matsuda}(1974)}]{matsuda1974}%
  \BibitemOpen
  \bibfield  {author} {\bibinfo {author} {\bibfnamefont {Hirotsugu}\
  \bibnamefont {Matsuda}},\ }\bibfield  {title} {\enquote {\bibinfo {title}
  {{Infinite Susceptibility without Spontaneous Magnetization: Exact Properties
  of the Ising Model on the Cayley Tree}},}\ }\href {\doibase
  10.1143/PTP.51.1053} {\bibfield  {journal} {\bibinfo  {journal} {Progress of
  Theoretical Physics}\ }\textbf {\bibinfo {volume} {51}},\ \bibinfo {pages}
  {1053--1063} (\bibinfo {year} {1974})}\BibitemShut {NoStop}%
\bibitem [{Note1()}]{Note1}%
  \BibitemOpen
  \bibinfo {note} {Sykes \protect \emph {et al}\cite {sykes66} make a
  distinction between a strong embedding and a weak embedding. A strong
  emebdding implies that all edges between vertices present in the graph are
  present in the cluster, whereas a weak embedding does not require that. The
  LCE we use relies on weak embdeddings.}\BibitemShut {Stop}%
\bibitem [{Note2()}]{Note2}%
  \BibitemOpen
  \bibinfo {note} {This results generalizes to the $q$-state Potts model as can
  be verified directly by computing $\protect \qopname \relax o{log}Z_N$ for
  $N=1,2,3,4,\protect \ldots $ and calculating the corresponding weights. By
  obtaining a formula for $\protect \qopname \relax o{log}Z_N$, we can show
  that $W_j=0$ for $j\geq 3$, just as in the spin-$\protect \nicefrac 12$
  classical Ising model}\BibitemShut {NoStop}%
\bibitem [{\citenamefont {M\"uller-Hartmann}\ and\ \citenamefont
  {Zittartz}(1974)}]{muller-hartmann1974}%
  \BibitemOpen
  \bibfield  {author} {\bibinfo {author} {\bibfnamefont {E.}~\bibnamefont
  {M\"uller-Hartmann}}\ and\ \bibinfo {author} {\bibfnamefont {J.}~\bibnamefont
  {Zittartz}},\ }\bibfield  {title} {\enquote {\bibinfo {title} {New type of
  phase transition},}\ }\href {\doibase 10.1103/PhysRevLett.33.893} {\bibfield
  {journal} {\bibinfo  {journal} {Phys. Rev. Lett.}\ }\textbf {\bibinfo
  {volume} {33}},\ \bibinfo {pages} {893--897} (\bibinfo {year}
  {1974})}\BibitemShut {NoStop}%
\end{thebibliography}%

\appendix
\section{Branched clusters at low field for $m=3$}
\label{appsec:branched}

In Section~\ref{sec:cayley-tree-1}, we noted that at lowest order in $h$, the branched clusters do
not contribute any weight. For completeness, we show a few examples of this.

\subsection{One branch chains}
\label{sec:one-branch-chains}

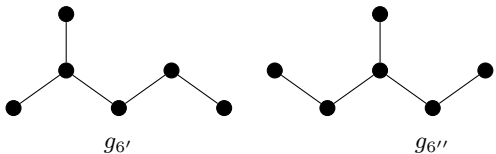
\begin{figure}[h]
  \centering
  \begin{tikzpicture}
    \node[shape=circle,draw=black,fill,inner sep=2pt] (a) at (0,0) {};
    \node[shape=circle,draw=black,fill,inner sep=2pt] (b) at (0.7,0.5)
    {}; \node[shape=circle,draw=black,fill,inner sep=2pt] (c) at
    (1.4,0) {}; \node[shape=circle,draw=black,fill,inner sep=2pt] (d)
    at (0.7,1.25) {}; \node[shape=circle,draw=black,fill,inner
    sep=2pt] (e) at (2.1,0.5)
    {};\node[shape=circle,draw=black,fill,inner sep=2pt] (f) at
    (2.8,0) {}; \path[-] (a) edge (b); \path[-] (b) edge (c); \path[-]
    (b) edge (d); \path[-] (c) edge (e); \path[-] (e) edge (f); \node
    (e) at (1.4,-0.5) {$g_{6'}$};
  \end{tikzpicture}\hspace{10pt}
  \begin{tikzpicture}
    \node[shape=circle,draw=black,fill,inner sep=2pt] (a) at (0,0) {};
    \node[shape=circle,draw=black,fill,inner sep=2pt] (b) at (0.7,0.5)
    {}; \node[shape=circle,draw=black,fill,inner sep=2pt] (c) at
    (1.4,0) {}; \node[shape=circle,draw=black,fill,inner sep=2pt] (d)
    at (0.7,1.25) {}; \node[shape=circle,draw=black,fill,inner
    sep=2pt] (e) at (2.1,0.5)
    {};\node[shape=circle,draw=black,fill,inner sep=2pt] (f) at
    (-0.7,0.5) {}; \path[-] (a) edge (b); \path[-] (b) edge (c);
    \path[-] (b) edge (d); \path[-] (c) edge (e); \path[-] (a) edge
    (f); \node (e) at (1.4,-0.5) {$g_{6''}$};
  \end{tikzpicture}
  \caption{6-vertex graphs with one branch}
  \label{fig:g5g6-branched}
\end{figure}

We first consider chains of all lengths that have one branched vertex
somewhere along the chain. Fig.~\ref{fig:g5g6-branched} shows the two
possibilities at $n=6$. Starting with the smallest branched cluster, $g_{4'}$,
we get for the free energy,
\begin{multline}
  \label{eq:31}
  -\beta F_{g_{4'}} = 4\log 2 + 3\log(\cosh K) + \\h^{2}\big[2 + 3\tanh K(1+\tanh K)\big] + O(h^{4}).
\end{multline}
The corresponding weight is given by
\begin{equation}
  \label{eq:32}
  W_{g_{4'}} = -\beta F_{g_{4'}} - 3W_{g_{3}} - 3W_{g_{2}} - 4W_{g_{1}} = O(h^{4}).
\end{equation}

Similarly, for $g_{5'}$, we get
\begin{multline}
  -\beta F_{g_{5'}} = 5\log 2+ 4\log(\cosh K) + \\h^{2}\Big[\frac52 + 2\tanh K \{2+\tanh K(2+\tanh K)\}\Big]
  + O(h^{4}).
\end{multline}
The corresponding weight is given by
\begin{multline}
  \label{eq:32}
  W_{g_{5'}} = -\beta F_{g_{5'}} - W_{g_{4'}} - 2W_{g_{4}} - 4W_{g_{3}} - 4W_{g_{2}} - 5W_{g_{1}} \\= O(h^{4}).
\end{multline}

This continues to remain true at higher orders, showing that
chains extending out on to one side do not
contribute to the low-magnetic field free energy. This leaves us with
graphs like $g_{6''}$ which, one can show, also have zero weight at this order.

We therefore conclude that at $O(h^{2})$ clusters with a single branch
do not contribute to the free energy and therefore the magnetization.

\subsection{Two branch chains}
\label{sec:two-branch-graphs}

A direct calculation of the partition function and the corresponding weights on clusters with two
branches reveals that their weights are also zero at $O(h^{2})$, leading us to conclude that
this is true for all branched chains.

\end{document}